\documentclass[12pt]{amsart}

\usepackage{amssymb}
\usepackage[centertags]{amsmath}

\usepackage{cite}

\newtheorem%
{thm}{Theorem}[section]
\newtheorem%
{proposition}[thm]{Proposition}
\newtheorem%
{lemma}[thm]{Lemma}
\newtheorem%
{definition}[thm]{Definition}
\newtheorem%
{corollary}[thm]{Corollary}
\newtheorem%
{conjecture}[thm]{Conjecture} \theoremstyle{definition}

\theoremstyle{remark}

\newcommand{\dontprint}[1]{\relax}

\usepackage[all]{xy}

\title
{Peierls brackets in non-Lagrangian field theory}
\author{A.A. Sharapov}
\address{Physics Faculty, Tomsk State University, Tomsk 634050, Russia}
\email{sharapov@phys.tsu.ru}

\begin{document}
\maketitle

\begin{abstract}
The concept of Lagrange structure allows one to systematically quantize the Lagrangian and non-Lagrangian dynamics within the path-integral approach. In this paper, I show that any Lagrange structure gives rise to a covariant Poisson brackets on the space of solutions to the classical equations of motion, be they Lagrangian or not.  The brackets generalize the well-known Peierls' bracket construction and make a bridge between the path-integral and the deformation quantization of non-Lagrangian dynamics.

\end{abstract}

\section{Introduction}

The least action principle provides the foundation for classical mechanics and field theory. A distinguishing feature of  the Lagrangian equations of motion among other differential equations is that their  solution space carries a natural symplectic structure, making it into a phase space. The physical observables, being identified with the smooth function(al)s on the phase space, are then endowed with the structure of a Poisson algebra. This algebraic formulation serves as a starting-point for the procedure of canonical quantization. The seminal Noether's theorem on the relationship between symmetries and conservation laws is an added reason in favour of the Lagrangian formalism.

In spite of its  indubitable elegance and power,  the least action principle does not meet all  the demands of modern high energy physics: There is a great deal of fundamental field-theoretical models whose equations of motion do not follow from the least action principle. An incomplete list of examples includes the self-dual Yang-Mills fields, Siberg-Witten and Donaldson-Uhlenbeck-Yau equations, equations describing 5-brane,  various superconformal theories with extended supersymmetry, Vasiliev's equations for massless higher-spin fields, etc.  In the absence of action, the standard quantization procedures - either operatorial or path-integral - are no longer applicable to the classical theory.

In \cite{KazLS}, a systematic method was proposed for the covariant quantization of Lagrangian and non-Lagrangian theories. In brief, the method allows one to define a path integral for the quantum averages of physical observables starting from the classical equations of motion. Central to this quantization method is the notion of a Lagrange structure. In most cases the existence of a Lagrange structure is less restrictive for the classical dynamics than the existence of an action functional. Furthermore, one and the same system of equations may admit a variety of different Lagrange structures leading to inequivalent quantizations.
Given a Lagrange structure, one can assign the configuration space of fields (also known as the space of all histories) with a probability amplitude $\Psi[\varphi]$ such that the quantum average of an observable $A[\varphi]$ is given by the path integral
\begin{equation}\label{psi}
\langle A\rangle=\int D\varphi A[\varphi]\Psi[\varphi]\,.
\end{equation}
For Lagrangian equations of motion endowed with the canonical Lagrange structure, the probability amplitude takes the standard Feynman's form $\Psi=e^{\frac i\hbar S}$, where $S[\varphi]$ is the action functional. In the general case, the probability amplitude cannot be represented as the exponential of a local functional.

The path-integral representation for the transition amplitudes is usually derived starting from the operatorial approach to quantum mechanics. The latter is considered as the most reliable way to produce the right integration measure that ensures unitarity. On the other hand, one can start from  a probability amplitude $\Psi$, be it of Feynman's  form or not, and ask about the operator algebra of quantum observables underlying the path integral (\ref{psi}). In the classical limit this algebra should reproduce the corresponding Poisson algebra of physical observables.  In the most straightforward way this correspondence between the classical and quantum algebras is realized in deformation quantization. Due to the famous Kontzevich's theorem \cite{K}, we know that the deformation quantization of the Poisson algebra of classical observables is a purely formal procedure with controllable ambiguity.
Thus, to reconstruct the algebra of quantum observables underlying the path integral (\ref{psi}) one should be able to identify, first and foremost,  the corresponding Poisson structure.

In the case of Lagrangian theories, there are at least two different ways to define the Poisson structure on the solution space. The first one is the standard Hamiltonian formalism, which requires an explicit splitting of space-time into space and time and introduction of canonical momenta. The main drawback of this approach is the lack of manifest covariance, which  causes some complications in applying it to relativistic field theory.  An alternative approach to the description of the Poisson algebra of physical observables was proposed by Peierls in his seminal 1952 paper \cite{P}. In that paper he invented what is now known as the Peierls brackets on the covariant phase space. In contrast to the usual (non-covariant) Hamiltonian formalism, where the phase space is identified with the space of initial data, the covariant phase space is the space of all solutions to the Lagrangian equations of motion. Peierls' paper opened up the way for constructing a fully relativistic theory of quantum fields \cite{DWbook}.   For more recent discussions of the Peierls brackets, on different levels of rigor,  I refer the reader to \cite{DW,M, BEMS, FR, Kh}.

The aim of this paper is to extend the covariant phase-space approach to the case of general (i.e., not necessarily Lagrangian) theories. More precisely, we will show that any Lagrange structure gives rise to a Poisson structure in the space of solutions to the classical equations of motion. The corresponding  Poisson brackets are fully covariant and reduce to the Peierls brackets in the case of Lagrangian theories endowed with the canonical Lagrange structure.
  It is pertinent to note that for mechanical systems described by ordinary differential equations, a relationship between the Lagrange and Poisson structures has been already established in \cite{KLS0}. The construction was somewhat indirect and required the equations of motion to be brought into the first-order normal form. In the present paper, we give an explicit formula for the covariant Poisson brackets, which directly applies to the general mechanical systems as well as field theories.

   Note that our exposition is mostly focused on the algebraic and geometric aspects of the construction, while more subtle functional analytical details are either ignored or treated in a formal way. These details, however, are not specific to our problem  and can be studied, in principle, along the same lines as in the case of the conventional Peierls' brackets.

The rest of the paper is organized as follows. Since the concept of a Lagrange structure is not widely known to non-experts, we briefly discuss it in Sec. 3. The exposition is based on the material of Sec. 2, where we recall the basic notions  and constructions concerning the classical field theory in non-Lagrangian setting. The covariant Poisson brackets on the space of physical observables appear in Sec. 4. In Sec. 5, we illustrate the general construction by three examples: the Pais-Uhlenbeck oscillator, chiral bosons in two dimensuions, and Maxwell's electrodynamics in the first-order formalism. In each case we present the covariant Poisson brackets of dynamical variables and evaluate their equal-time limit.  In Sec. 6, we summarize our results and make some comments on the deformation quantization of the covariant Poisson brackets.

\section{Classical gauge systems}

\subsection{Kinematics} In modern language the  classical fields are just the sections of a locally trivial fiber bundle $B\rightarrow M$ over the space-time manifold $M$. The typical fiber $F$ of $B$ is called the target space of fields. In case the bundle is trivial, i.e., $B=M\times F$, the fields are merely  the mappings from $M$ to $F$. In each trivializing coordinate chart $U\subset M$ a  field $\varphi: M\rightarrow B$ is described by a collection of functions $\varphi^i(x)$, where $x\in U$ and $\varphi^i$ are local coordinates in $F$. These functions are often called the components of the field $\varphi$.

Formally, one can think of $\Gamma(B)$ -- the space of all field configurations -- as a smooth manifold $\mathcal{M}$ with the continuum infinity of dimensions and $\varphi^i(x)$ playing the role of local coordinates. In other words, the different local coordinates $\varphi^i(x)$ on $\mathcal{M}$ are labeled by the space-time point $x\in M$ and the discrete index $i$. To emphasize this interpretation of fields as coordinates on the infinite-dimensional manifold $\mathcal{M}$ we will include the space-time point $x$ into the discrete index $i$ and write $\varphi^i$ for $\varphi^i(x)$; in so doing, the summation over the ``superindex'' $i$ implies usual summation for its discrete part and integration over $M$ for $x$. In the physical literature this convention is known as DeWitt's condensed notation \cite{DWbook}.

Proceeding with the infinite-dimensional geometry above,  we identify  the ``smooth functions'' on the ``manifold'' $\mathcal{M}$ with the infinitely differentiable functionals of field $\varphi$. These functionals form a commutative algebra, which will be denoted by $\Phi$. If $\delta\varphi^i$ is an infinitesimal variation of field, then, according to the condensed notation, the corresponding variation of a functional $S\in \Phi$ can be written in the form
\begin{equation}\label{dS}
\delta S=S,_i\delta\varphi^i\,,
\end{equation}
where the comma denotes the functional derivative.

The concepts of vector fields, differential forms and exterior differentiation on $\mathcal{M}$ are naturally introduced through the functional derivatives, see e.g. \cite{KLS1}. In particular,  the variations $\delta \varphi^i$ span the space of 1-forms and the functional derivatives
$\delta/\delta\varphi^i$ define a basis in the tangent space  $T_\varphi \mathcal{M}$. So, we can speak of the tangent and cotangent bundles of $\mathcal{M}$.

The tangent and cotangent bundles are not the only vector bundles that can be defined over $\mathcal{M}$. Given a vector bundle $E\rightarrow M$ over the space-time manifold, we define the vector bundle $\mathcal{E}\rightarrow \mathcal{M}$ whose sections a smooth functionals of fields with values in $\Gamma(E)$. In other words, a section $\xi\in \Gamma(\mathcal{E})$ takes each field configuration $\varphi \in \mathcal{M}$ to a section $\xi[\varphi]\in \Gamma(E)$. Here we do not require the section $\xi[\varphi]$ to be smooth; discontinuous or even  distributional sections are also allowed. We will refer to $\mathcal{E}$ as the vector bundle associated with $E$. The dual vector bundle $\mathcal{E}^\ast$ is defined to be the vector bundle associated with $E^\ast$.

Of course, care must be exercised when extending the standard differential-geometric constructions to the infinite-dimensional setting. One should keep in mind that according to the condensed notation the contractions $A_a B^a$ of dual sections involve integration over the space-time, and hance, the result may be ill-defined unless appropriate support conditions are imposed on the sections contracted. In particular, associativity
$$
(A^a B_a^b)C_b=A^a(B_a^bC_b)
$$
is ensured only when the various space-time integrals converge properly.

In order to control convergence as well as to justify our subsequent constructions some restrictions are to be imposed on the structure of the underlying space-time manifold. In this paper, our basic assumption will be that $M$ is a globally hyperbolic manifold endowed with a volume form. In the most of field-theoretical models both the structures come from or compatible with a Lorentzian metric on $M$. The globally hyperbolic manifolds is a natural arena  for  the theory of hyperbolic differential equations with well-posed Cauchy problem. By definition, each  globally hyperbolic manifold  $M$ admits a global time function whose level surfaces provide a foliation of $M$ into space-like Cauchy surfaces $N$, so that $M\simeq \mathbb{R}\times N$. Using the direct product structure, one can cut $M$ into the ``past'' and the ``future'' with respect to a given instant of time $t\in \mathbb{R}$:
$$
M^-_t=(-\infty, t]\times N\,,\qquad M^+_t=[t,\infty )\times N\,,\qquad M=M_t^-\cup M_t^+\,.
$$

Given a vector bundle $E\rightarrow M$, we define the following subspaces in the space of sections $\Gamma(E)$:

\begin{itemize}
  \item $\Gamma_0(E)=\{\xi\in \Gamma(E)\,|\, \mbox{supp $\xi$ is compact}\}$;
  \item $\Gamma_{sc}(E)=\{\xi \in \Gamma(E)\,|\, \mbox{supp $\xi$ is spatially compact}\}$;
  \item $\Gamma_-(E)=\{\xi\in \Gamma_{sc}(E)\,|\, \mbox{supp $\xi \subset M_t^-$ for some $t$}\}$;
  \item $\Gamma_+(E)=\{\xi \in\Gamma_{sc}(E)\,|\,\mbox{supp $\xi\subset M^+_t$ for some $t$}\}$.
\end{itemize}
Here the spatially compact means that the intersection $M_t^-\cap \mathrm{supp} \,\xi \cap M_{t'}^+$ is compact for any $t\geq t'$.  We will refer to the elements of $\Gamma_-(E)$ and $\Gamma_+(E)$ as the sections with retarded and advanced support, respectively. When checking the convergence of various integrals we will deal with below the following property is of particular assistance: If $\xi_1\in \Gamma_-(E_1)$ and $\xi_2\in \Gamma_+(E_2)$, then $\xi_1\xi_2\in \Gamma_0(E_1\otimes E_2)$, whenever the product of the (distributional) sections $\xi_1$ and $\xi_2$ is well-defined.

A differentiable functional $A$ is said to be compactly supported if $A,_{i}\in \Gamma_0(T^\ast\mathcal{M})$. For example, a local functional, like the action functional, is compactly supported if it is given by an integral over a compact  domain. It is clear that the formally smooth and compactly supported functionals form an $\mathbb{R}$-algebra with respect to the usual addition and multiplication. We will denote this algebra by $\Phi_0$.
Let now $\mathcal{E}$ be a vector bundle associated with $E$. We say that a section $\xi \in \Gamma(\mathcal{E})$ has retarded, advanced or compact support if $\xi[\varphi]\in \Gamma(E)$ does so for any field configuration $\varphi\in \mathcal{M}$. The sections with the mentioned support properties  form
subspaces in $\Gamma(\mathcal{E})$, which will be denoted by $\Gamma_-(\mathcal{E})$, $\Gamma_+(\mathcal{E})$, and $\Gamma_0(\mathcal{E})$, respectively.

When dealing with local field theories it is also useful to introduce the subspace of local sections $\Gamma_{loc}(\mathcal{E})\subset \Gamma(\mathcal{E})$. This consists of those   sections of $E$ whose components are given, in each coordinate chart, by smooth functions of the field $\varphi$ and its partial derivatives up to some finite order. For instance, the Euler-Lagrange equations $S,_i$ for the action $S$ constitute a section of $\Gamma_{loc}(T^\ast \mathcal{\mathcal{\mathcal{M}}})$.

\subsection{Dynamics}The dynamics of fields are specified by a set of differential equations
\begin{equation}\label{T}
T_a[\varphi]=0\,.
\end{equation}
Here $a$ is to be understood as including a space-time point. According to our definitions, the left hand sides of the equations can be viewed  as components of a local section of some vector bundle $\mathcal{E}$ over $\mathcal{M}$. We call $\mathcal{E}$ the \textit{dynamics bundle}. Since we do not assume the field equations (\ref{T}) to come from the least action principle, the discrete part of the condensed index $a$ may have nothing to do with that of $i$ labeling the field components.
 In the special case of Lagrangian systems the dynamics bundle coincides with the cotangent bundle $T^\ast \mathcal{M}$ and the field equations  are determined by the exact 1-form  (\ref{dS}), with $S$ being the action functional.

Let $\Sigma$ denote the space of all solutions to the field equations (\ref{T}). Geometrically, we can think of $\Sigma$ as a smooth submanifold of $\mathcal{M}$ and refer to $\Sigma$ as the \textit{dynamical shell} or just the \textit{shell}. For the Lagrangian systems the shell is just the set of all stationary points of the action $S$. By referring  to  $\Sigma$ as a smooth submanifold we mean that the standard regularity conditions hold for the field equations. These conditions can be formulated as follows \cite{H}, \cite{BBH}. For any integer $n$ we introduce the space $J^nB$ of  $n$-jets of the field $\varphi$. By definition, $J^nB$ is a smooth manifold (and even fiber bundle) with local coordinates given by the space-time coordinates $x^\mu$ and the partial derivatives of the field $\varphi^i(x)$ up to the $n$-th order. If $n$ is the highest order of derivative occurring in (\ref{T}), then the field equations and their differential consequences \begin{equation}\label{RC}
T_a=0\,,\qquad \partial_\mu T_a=0\,, \quad \ldots\,,\quad \partial_{\mu_1}\ldots \partial_{\mu_k}T_a=0\,,
 \end{equation}
 being regraded as algebraic equations on jets,  define a surface in $J^{n+k}B$. We will assume that for any $k$ the equations (\ref{RC}) define a smooth surface indeed and provide a regular representation of that surface\footnote{
The algebraic equations  $F_a=0$ are said to provide a regular representation of a surface $\mathcal{S}$ if one can locally split the functions $F_a$ into independent functions $F_{{\overline{a}}}$ and dependent functions $F_{\underline{a}}$ in such a way that (i) $\mathcal{S}$ is fully determined by the equations $F_{\overline{a}}=0$ and (ii) the covectors $dF_{\overline{a}}$ are linearly independent on $\mathcal{S}$.}. In most theories of physical interest the regularity conditions are fulfilled at least locally.

 Given the shell, a functional $A\in \Phi_0$ is said to be trivial iff $A|_\Sigma=0$. Clearly, the trivial functionals form an ideal of the algebra $\Phi_0$. Denoting this ideal by  $\Phi_0^{\mathrm{triv}}$, we define the quotient-algebra $\Phi_0^\Sigma=\Phi_0/\Phi_0^{\mathrm{triv}}$. The regularity conditions above imply that for each trivial functional $A\in \Phi_0^{\mathrm{triv}}$  there exists a (distributional) section $\xi\in \Gamma(\mathcal{E}^\ast)$ such that $A=\xi^a T_a$.
 In other words, the trivial functionals are precisely those that are proportional to the equations of motion and their differential consequences. By definition, the elements of the algebra  $\Phi_0^\Sigma$ are given by the equivalence classes of functionals from $\Phi_0$, where two functionals $A$ and $B$ are considered to be equivalent if $A-B\in \Phi_0^{\mathrm{triv}}$. In that case we will write $A\approx B$. Formally, one can think of $\Phi_0^\Sigma$ as the space of smooth, compactly supported functionals on $\Sigma$.

The functional derivative of the field equations (\ref{T}) produces what is known as the Jacobi operator
\begin{equation}\label{J}
J_{ai}=T_{a,i}\,.
\end{equation}
In general, the functional derivatives of the section $T\in \Gamma_{loc}(\mathcal{E})$ transform inhomogeneously under the bundle automorphisms, so that the values $T_{a,i}$ do not constitute a section of $\mathcal{E}\otimes T^\ast \mathcal{M}$.  The interpretation of $J$ as a globally defined section can be restored by choosing a linear connection on $\mathcal{E}$ and replacing the partial functional derivatives by the covariant ones as, for example, in \cite{KazLS}, \cite{LS2}. In this paper, however, we choose not to follow this approach. To simplify our exposition we assume that all vector bundles over $\mathcal{M}$ we deal with are trivial and the partial functional derivatives are just the covariant derivatives associated with flat connection. Actually, the covariant Poisson brackets, that will be construct in Sec. 4,  are independent of the choice of connection, so the restriction imposed on the global structure of vector bundles is purely technical\footnote{Notice that the restriction of $J$ onto the solution space $\Sigma$ is well-defined. It is the operator $J[{\varphi}]$ that defines the linearization of the equations of motion (\ref{T}) around a given solution $\varphi\in \Sigma$.}.

 As we are dealing with local field theory, $J[\varphi]$ represents the integral kernel of a differential operator acting from $\Gamma(TM)$ to $\Gamma(E)$ with coefficients depending on $\varphi$ and its derivative up to some finite order. As the differential operators do not increase the supports of sections they act upon, one may be sure that the section $J_{ai}V^i$ of the dynamics bundle $\mathcal{E}$ has advanced, retarded or compact support if the vector field $V\in \Gamma(T\mathcal{M})$ does so. This simple observation will be of frequent use in our subsequent considerations.

\subsection{Gauge symmetries and identities} The field equations (\ref{T}) are said to be \textit{gauge invariant} if there exist a vector bundle $\mathcal{F}\rightarrow \mathcal{M}$ together with a section $R=\{R_\alpha^i\}$ of $\mathcal{F}^\ast\otimes T\mathcal{M}$ such that
\begin{equation}\label{JR}
J_{ai}R^i_\alpha \approx 0\,.
\end{equation}
In local field theory it is also assumed that $R_\alpha^i[\varphi]$ is the integral kernel of  a differential operator $R[\varphi]: \Gamma({\mathcal{F}})\rightarrow \Gamma(T{M})$ for each $\varphi \in \mathcal{M}$.

Since the bundle $\mathcal{F}$ is assumed to be trivial, we can think of $R_\alpha=\{R_\alpha^i\}$ as a collection of vector fields on $\mathcal{M}$.  This vector fields are called the gauge symmetry generators. The terminology is justified by the fact that for any infinitesimal section $\varepsilon\in \Gamma_0(\mathcal{F})$ the infinitesimal change of field $\varphi^i\rightarrow \varphi^i +\delta_\varepsilon \varphi^i$, where $$
\delta_{\varepsilon}\varphi^i=R^i_\alpha\varepsilon^\alpha\,,
$$
 maps solutions of (\ref{T}) to solutions.  In other words, the vector fields $R_\alpha$ are tangent to the dynamical shell $\Sigma$. The gauge symmetry transformations are said to be trivial if $R\approx 0$.
 If the vector bundle $\mathcal{F}$ is big enough to accommodate all nontrivial gauge symmetries, then we call $\mathcal{F}$ the \textit{gauge algebra bundle} and refer to $R_\alpha$ as a complete set of gauge symmetry generators.  It follows from the definition (\ref{JR}) that the vector fields $R_\alpha$ define an on-shell involutive vector distribution on $\mathcal{M}$, i.e.,
 $$
 [R_\alpha, R_\beta]\approx C_{\alpha\beta}^\gamma R_\gamma\,,
 $$
 for some $C$'s. This distribution will be denoted by $\mathcal{R}$ and called \textit{gauge distribution}.

A functional $A\in \Phi_0$ is gauge invariant if
$$
A,_iR^i_\alpha\approx 0\,.
$$
In that case we say that $A$ represents a \textit{physical observable}. The gauge invariant functionals form a subalgebra $\Phi_0^{\mathrm{inv}}$ in $\Phi_0$. Two gauge invariant functionals $A$ and $A'$ are considered as equivalent or represent the same physical observable if $A\approx A'$. So, we identify the physical observables
with the equivalence classes of gauge invariant functionals from $\Phi_0$. This definition is consistent as the trivial functionals are automatically gauge invariant and the property of being gauge invariant passes through the quotient $\Phi^{\mathrm{inv}}_0/\Phi_0^{\mathrm{triv}}$. The physical observables form a commutative algebra, which will be denoted by $\mathcal{O}$. In what follows we will very often identify physical observables with their particular representatives in $\Phi^{\mathrm{inv}}_0$.

In general, it may be impossible to choose a complete set of gauge generators $R_\alpha$ in a linearly independent way. In other wards, any complete set may happen to be overcomplete, meaning the existence of nontrivial linear relations among the vector fields $R_\alpha \in \mathcal{R}$:
\begin{equation}\label{RRR}
R_{\alpha_1}^\alpha R_\alpha^i\approx 0\,.
\end{equation}
As above, $R_1=\{R_{\alpha_1}^\alpha\}$ is a section of an appropriate vector bundle over $\mathcal{M}$ coming from  the kernel of a differential operator. If  $R_1$ does not vanish on shell, then one says that the gauge symmetry generated by $R_\alpha$ is \textit{reducible}. Accordingly, $R_{\alpha_1}=\{R_{\alpha_1}^\alpha\}$ are called the generators of reducibility relations (\ref{RRR}). It is well possible that the  generators $R_{\alpha_1}$ form an overcomplete basis  in the space of solutions to the linear system (\ref{RRR}), in which case we should  consider the  reducibility relations for the generators $R_{\alpha_1}$, and so on.  Proceeding in this way we finally arrive at a sequence of reducibility relations generated by the differential operators $\{R_{\alpha_{i+1}}^{\alpha_i}\}_{i=0}^n$ with the property $R_{\alpha_{i+1}}^{\alpha_i} R_{\alpha_{i}}^{\alpha_{i-1}}\approx 0$. The number of terms in the sequence, $n$, is called the order of reducibility. We will allow the order of reducibility to be arbitrary large or even infinite.

We say that the field equations (\ref{T}) are \textit{gauge dependent} or admit gauge identities, if there exists a vector bundle $\mathcal{G}\rightarrow \mathcal{M}$ together with a section $L=\{L_A^a\}$ of $\mathcal{E}^\ast\otimes \mathcal{G}$ such that
\begin{equation}\label{LT}
L^a_AT_a\equiv 0\,.
\end{equation}
Again, in local field theory the section $L$ is assumed to be given by the integral kernel of a differential operator.  Its components $L_A=\{L_A^a\}$ are called the generators of gauge identities. Varying (\ref{LT}) by $\varphi^i$, we get
\begin{equation}\label{LJ}
L^a_AJ_{ai}\approx 0\,.
\end{equation}
The last relation provides the on-shell definition of gauge dependence in terms of the Jacobi operator.   Like (\ref{JR}), Eq. (\ref{LJ}) admits a plenty of trivial solutions. Namely, a generator of gauge identity is called trivial if $L_A\approx 0$.   This suggests to consider the equivalence classes of gauge identities modulo trivial ones. Then, one can see that for regular field equations the equivalence class of each solution to (\ref{LJ}) contains a representative satisfying (\ref{LT}). So, the off- and on-shell definitions (\ref{LT}) and (\ref{LJ}) are essentially equivalent.  If the bundle $\mathcal{G}$ is big enough to accommodate all the nontrivial gauge identities, we refer to it as the \textit{bundle of gauge identities}.

All that have been said above about reducibility of gauge symmetries can be literally  repeated for the gauge identities (\ref{LJ}): The generators $L_A$ may happen to be reducible, giving rise to a sequence of reducibility relations generated by the differential operators $\{L^{A_i}_{A_{i+1}}\}_{i=0}^m$ with the property $L_{A_{i+1}}^{A_i}L_{A_i}^{A_{i-1}}\approx 0$.

The information about the  gauge symmetries and identities can be compactly encoded by the following diagram:
\begin{equation}\label{D1}
\xymatrix@C=0.5cm{
  \cdots \ar[r] & \Gamma(\mathcal{F}) \ar[rr]^{R} && \Gamma(T\mathcal{M}) \ar[rr]^{J} && \Gamma(\mathcal{E}) \ar[rr]^{L} && \Gamma(\mathcal{G}) \ar[r] & \cdots }
\end{equation}
Here the maps $R$, $J$, and $L$ are defined, respective, by the generators of gauge symmetry, the Jacobi operator, and the generators of gauge identities. The dots stand for the chains of possible reducibility relations. As the presence of reducibility relations is inessential for our subsequent consideration, we do not indicate the corresponding vector bundles and the maps explicitly. The missed data, however, are necessary for constructing the BRST description of the gauge dynamics, see \cite{KazLS},\cite{KLS2}.

The main property of the diagram (\ref{D1}) is that it makes a co-chain complex upon restriction to the shell, so that the composite of any two consecutive maps is zero:
$$
\ldots,\quad (J\circ R)|_\Sigma=0\,,\qquad (L\circ J)|_\Sigma=0\,,\quad\ldots
$$
We will denote this complex by $\mathcal{C}$.
The complex $\mathcal{C}$ contains several subcomplexes of physical interest. These are obtained by imposing certain restrictions on the spaces of sections forming the complex. First, restricting  the maps (\ref{D1}) to the subspaces of local sections, that is, replacing $\Gamma \rightarrow \Gamma_{loc}$, we get the on-shell complex $\mathcal{C}_{loc}\subset \mathcal{C}$. It is clear that the co-boundary operators, being differential operators with local coefficients, take the local sections to local ones.  As the following isomorphism illustrates, see e.g. \cite{KLS2},  the cohomology groups of the complex $\mathcal{C}_{loc}$ may well be nontrivial:
$$
\left(\frac{\mathrm{Ker}\, J}{\mathrm{Im}\, R}\right)_{\mathcal{C}_{loc}}\simeq (\mbox{The space of global symmetries})\,.
$$

Using the fact that the differential operators do not increase the supports of sections which they act upon, three more complexes can be defined, namely, the complexes $\mathcal{C}_\pm$, $\mathcal{C}_0$ composed of the spaces of sections with advanced, retarded  or compact supports. Our main assumption  will be that the complexes $\mathcal{C}_\pm$ are acyclic, so that we have two exact sequences\footnote{By abuse of notation we use the same letters for the maps in (\ref{D1}) and their restrictions in (\ref{ExSeq}).}
\begin{equation}\label{ExSeq}
\xymatrix@C=0.5cm{
  \cdots \ar[r] & \Gamma_{\pm}(\mathcal{F}|_\Sigma) \ar[rr]^{R} && \Gamma_{\pm}(T\mathcal{M}|_\Sigma) \ar[rr]^{J} && \Gamma_{\pm}(\mathcal{E}|_{\Sigma}) \ar[rr]^{L} && \Gamma_{\pm}(\mathcal{G}|_{\Sigma}) \ar[r] & \cdots }
\end{equation}
 The adjective ``exact'' means that the image of each map coincides exactly with the kernel of the next one, i.e.,
$$
\ldots,\quad \mathrm{Im}\, R=\mathrm{Ker}\, J\,,\qquad \mathrm{Im}\, J=\mathrm{Ker}\, L\,,\quad \ldots
$$
This property admits the following interpretation. The linearization of the field equations (\ref{T}) around a given solution $\varphi_0 \in \Sigma$ gives the  linear homogeneous equations
\begin{equation}\label{JF}
J[\varphi_0]\phi=0\,.
\end{equation}
These equations are clearly invariant with respect to the gauge transformations
$$
\phi\quad \rightarrow\quad \phi' = \phi +R[\varphi_0]\varepsilon\qquad \varepsilon \in \Gamma(\mathcal{F})\,.
$$
Then  exactness at $\Gamma_{\pm}(T\mathcal{M}|_\Sigma)$ means that any solution to (\ref{JF}) that vanishes in the remote past/future is gauge equivalent to the zero one, that is, there exists $\varepsilon \in \Gamma_{\pm}(\mathcal{F})$ such that $\phi=R[\varphi_0]\varepsilon$. The last property is usually considered as the evidence of completeness of the gauge symmetry generators. The same dynamical interpretation in terms of linear homogeneous system of PDEs applies to the other operators making the on-shell exact sequences (\ref{ExSeq}).

Passing to the dual vector bundles in (\ref{D1}) and transposing the corresponding maps, we arrive at  the diagram
\begin{equation}\label{D2}
\xymatrix@C=0.5cm{
  \cdots \ar[r] & \Gamma(\mathcal{G}^\ast) \ar[rr]^{L^\ast} && \Gamma(\mathcal{E}^\ast ) \ar[rr]^{J^\ast} && \Gamma(T^\ast \mathcal{M}) \ar[rr]^{R^\ast} && \Gamma(\mathcal{\mathcal{F}^\ast}) \ar[r] & \cdots}
\end{equation}
whose restriction to $\Sigma$ gives one more complex, the dual complex $\mathcal{C}^\ast$. By definition, the maps $R^\ast$, $J^\ast$, and  $L^\ast$ are given by the integral kernel of the formally adjoint differential operators. Restricting then the supports of all sections,  we define the subcomplexes $\mathcal{C}_\pm^\ast\subset \mathcal{C}^\ast$. Again, we will assume that both the complexes $\mathcal{C}^\ast_{\pm}$ are acyclic, or what is the same, that the sequences
\begin{equation}\label{ExSeq*}
\xymatrix@C=0.5cm{
  \cdots \ar[r] & \Gamma_{\pm}(\mathcal{G}^\ast|_\Sigma) \ar[rr]^{L^\ast} && \Gamma_\pm(\mathcal{E}^\ast|_\Sigma ) \ar[rr]^{J^\ast} && \Gamma_\pm(T^\ast \mathcal{M}|_\Sigma) \ar[rr]^{R^\ast} && \Gamma_\pm(\mathcal{\mathcal{F}^\ast}|_\Sigma) \ar[r] & \cdots }
\end{equation}
are exact.

 For Lagrangian dynamics the structure of the diagrams (\ref{D1}) and (\ref{D2})  greatly simplifies due to the various identifications one can make in this case. Indeed, the Lagrangian equations are given by the functional derivatives of an action functional $S$, that is, constitute an exact 1-form $S,_i$ on $\mathcal{M}$. Hence, the dynamics bundle $\mathcal{E}$ of any Lagrangian theory is given by the cotangent bundle $T^\ast \mathcal{M}$.  The Jacobi operator, being given by the second functional derivatives of the action\footnote{The operator $S,_{ij}$ is also known as the Van Vleck matrix.}, $J_{ij}=S,_{ij}$, defines a linear map from the space of vector fields to the space of 1-forms on $\mathcal{M}$.  Due to the commutativity of  functional derivatives, the Jacobi operator is formally self-adjoint,  $J^\ast=J$, so that  one can always choose $\mathcal{G}=\mathcal{F}^\ast$ and set $L=R^\ast$. The last relation is a compact formulation of the second Noether's theorem on the one-to-one correspondence between the gauge symmetries and the gauge (or Noether) identities. This correspondence further extends to the reducibility relations resulting in the following diagram:
$$
\xymatrix@C=0.5cm{
  \cdots \ar[r] & \Gamma(\mathcal{F}) \ar[rr]^{R} && \Gamma(T\mathcal{M}) \ar[rr]^{J} && \Gamma(T^\ast\mathcal{M}) \ar[rr]^{R^\ast} && \Gamma(\mathcal{F}^\ast) \ar[r] & \cdots }
$$
The diagram  is formally self-dual and so is the corresponding on-shell complex, $\mathcal{C}=\mathcal{C}^\ast$.

\section{The Lagrange structure}\label{2}

 According to our definitions each classical field theory is completely specified by a pair  $(\mathcal{E}, T)$, where $\mathcal{E}\rightarrow \mathcal{M}$ is a vector bundle over the configuration space of fields and $T$ is a particular section of $\Gamma_{loc}(\mathcal{E})$. The solution space   $\Sigma$ is then identified with zero locus of the section $T$.
 Whereas the classical equations of motion $T_a[\varphi]=0$ are enough to formulate the classical dynamics they are certainly insufficient for constructing a quantum-mechanical description of fields. Any quantization procedure
has to involve one or another additional structure. Within the path-integral
quantization, for instance, it is the action functional that plays the role of such an extra
structure. The procedure of canonical quantization relies on the Hamiltonian form of dynamics,
involving a non-degenerate Poisson bracket and a Hamiltonian. Either approach assumes the existence of a variational formulation for the classical equations of motion (the least action principle)
and becomes inapplicable beyond the scope of variational dynamics. The extension of these quantization methods to non-variational dynamics was proposed in \cite{KazLS}, \cite{LS0}.
In particular, the least action principle of the Lagrangian formalism was shown to admit a far-reaching  generalization based  on the concept of a \textit{Lagrange structure}.

Like many fundamental concepts, the notion of a Lagrange structure can be introduced and motivated from various  perspectives. Some of these motivations and interpretations can be found in Refs. \cite{KazLS}, \cite{LS1}, \cite{KLS1}.
For our present purposes it is convenient to define the Lagrange structure as a collection of linear operators  $V, U, \ldots$  making the on-shell commutative diagram
\begin{equation}\label{LD}
\begin{array}{c}
\xymatrix@C=0.5cm{
  \cdots \ar[r] & \Gamma(\mathcal{F}) \ar[rr]^{R} && \Gamma(T\mathcal{M}) \ar[rr]^{J} && \Gamma(\mathcal{E}) \ar[rr]^{L} && \Gamma(\mathcal{G}) \ar[r] & \cdots\\
  \cdots \ar[r] & \Gamma(\mathcal{G}^\ast)\ar[u]_{U} \ar[rr]^{L^\ast} && \Gamma(\mathcal{E}^\ast) \ar[u]_{V}\ar[rr]^{J^\ast} && \Gamma(T^\ast\mathcal{M})\ar[u]_{V^\ast} \ar[rr]^{R^\ast} && \Gamma(\mathcal{\mathcal{F}^\ast}) \ar[u]_{U^\ast}\ar[r] & \cdots
  }
\end{array}
\end{equation}
Upon restriction to $\Sigma$ the vertical maps induce a morphism $\mathcal{C}^\ast\rightarrow \mathcal{C}$ of the on-shell complexes, which passes further to the cohomology.

The most important among the vertical maps is the operator $V$.  In \cite{KazLS}, it was given a special name \textit{Lagrange anchor}.
The defining property of the Lagrange anchor is the on-shell commutativity of the central square,
\begin{equation}\label{LAD}
J\circ V\approx V^\ast\circ J^\ast\,.
\end{equation}
Due to the regularity condition, the off-shell form of the last equality reads
\begin{equation}\label{LAD1}
J_{ai}V^i_b-V_a^iJ_{bi}=C_{ab}^d T_d
\end{equation}
for some $C$'s.  Setting
\begin{equation}\label{VV}
G=J\circ V\,,
 \end{equation}
we see that relation (\ref{LAD}) is equivalent to the formal on-shell self-adjointness of the operator $G: \Gamma(\mathcal{E}^\ast)\rightarrow \Gamma(\mathcal{E})$, that is, $G^\ast\approx G$. We call $G=\{G_{ab}\}$  the generalized \textit{Van Vleck operator}.

Given a Lagrange anchor $V$, the other vertical maps in (\ref{LD}) can be systematically reconstructed up to some trivial ambiguity. The proof of the last fact is an easy exercise in diagram chasing.     Indeed, letting $\eta=L^\ast\xi$ for some $\xi\in \Gamma(\mathcal{G}^\ast)$, we can write
$$
(J\circ V\circ) \eta\approx (V^\ast\circ J^\ast)\eta =(V^\ast\circ J^\ast\circ L^\ast)\xi \approx 0\,.
$$
Hence, the vector field $(V\circ L^\ast)\xi$ generates a gauge symmetry of the field equations, with $\xi\in \Gamma(\mathcal{G}^\ast)$ being the gauge parameter. Since the vector fields $R_\alpha $ are assumed to form a complete set of gauge generators, there must exist an operator $U: \Gamma(\mathcal{G}^\ast) \rightarrow \Gamma(\mathcal{F})$ such that
\begin{equation}\label{VL=RW}
V\circ L^\ast \approx R\circ U\,.
\end{equation}
Clearly, the last relation does not specify the operator $U$ completely as we are free to add to $U$ any operator that vanishes on shell or whose image belongs to the on-shell kernel of $R$.
Taking into account both the possibilities, we can write the most general solution to (\ref{VL=RW}) in the form
$$
 U^\alpha_A=U'^\alpha_A+T_a B^{a\alpha}_A+R^\alpha_{\alpha_1}E^{\alpha_1}_A \,.
$$
Here $U'$ is a particular solution to (\ref{VL=RW}), $R_{\alpha_1}=\{R_{\alpha_1}^\alpha\}$ are generators of the reducibility relations (\ref{RRR}), and the coefficients $B$ and $E$ are arbitrary sections of appropriate vector bundles.  By picking a particular operator $U$ one can then repeat  the reasoning above to construct the next in order vertical map (again with a controllable ambiguity) and so on.
The homotopy-theoretic interpretation of the arising ambiguity can be found in \cite{KLS1}.

Notice that for Lagranian dynamics,  the two exact sequences entering the diagram (\ref{LD}) are formally self-dual, and hence coincides with each other. In that case, we can take the upward arrows to define the identical linear maps. For $V=1$  the  condition (\ref{LAD1}) is satisfied due to the commutativity of the second functional derivatives, $J_{ij}=S,_{ij}=J_{ji}$ and  the Van Vleck operator $G$ coincides with the Jacobi operator $J$. The operator $V=1$ is referred to as the \textit{canonical Lagrange anchor} for variational equations of motion. It should be noted that even for the variational equations $S,_i=0$ there may exist non-canonical Lagrange anchors (any bi-Hamiltonian system is an example).

As with the generators of gauge symmetries, we can think of the Lagrange anchor as a collection of vector fields $V_a=\{V_a^i\}$ on $\mathcal{M}$. These generate a (singular) vector distribution $\mathcal{V}$, which we call the \textit{anchor distribution}. From the physical standpoint,  $\mathcal{V}$ defines the possible directions of  quantum fluctuations on $\mathcal{M}$. For Lagrangian theories endowed with the canonical Lagrange anchor $V=1$ all directions are allowable and equivalent. At the other extreme we have zero Lagrange anchor, $V=0$, for which the corresponding quantum system remains pure classical (no quantum fluctuations). In the intermediate situation only a part of physical degrees of freedom may fluctuate and/or the intensity of  fluctuations around a particular field configuration $\varphi\in \mathcal{M}$ may vary with $\varphi$.

Unlike the gauge distribution $\mathcal{R}$, the anchor distribution $\mathcal{V}$ is not generally involutive even on shell. The following two lemmas describe the differential properties of both the distributions, which will be used in the next section.

\begin{lemma}
The following commutation relations take place:
\begin{equation}\label{L1}
[R_\alpha, V_a]^i \approx C_{\alpha a}^bV^i_b+D_{\alpha a}^\beta R^i_\beta+J_{aj}W^{ji}_\alpha\,,
\end{equation}
where $W^{ji}_\alpha=W^{ij}_\alpha$ and the coefficients $C_{\alpha a}^b$ are defined by the relation
$$
R_\alpha T_a=C_{\alpha a}^bT_b\,.
$$
\end{lemma}

\begin{proof} Acting by the vector fields $R_\alpha$ on both the sides of relation (\ref{LAD1}), we get
\begin{equation}\label{RV}
\begin{array}{c}
0\approx R_\alpha(V_aT_b-V_bT_a)=[R_\alpha, V_a]T_b-[R_\alpha,V_b]T_a+V_aR_\alpha T_b-V_bR_\alpha T_a\\[3mm]
\approx[R_\alpha, V_a]T_b-[R_\alpha,V_b]T_a+C_{\alpha b}^cV_aT_c- C_{\alpha a}^cV_b T_c
\\[3mm]
\approx[R_\alpha, V_a]T_b-[R_\alpha,V_b]T_a+C_{\alpha b}^cV_cT_a- C_{\alpha a}^cV_c T_b\,.
\end{array}
\end{equation}
Introduce the vector fields
\begin{equation}\label{X}
X_{\alpha a}=[R_\alpha,V_a]-C_{\alpha a}^cV_c\,.
\end{equation}
Then (\ref{RV}) takes the form
$$
X_{\alpha a}T_b-X_{\alpha b}T_a\approx 0\,.
$$
In view of the regularity conditions, the general solution to this equation reads
$$
X_{\alpha a}^i \approx D_{\alpha a}^\beta R^i_\beta+J_{aj}W_\alpha^{ji}
$$
for some $D_{\alpha a}^\beta$ and $W_\alpha^{ij}=W_\alpha^{ji}$. Combining the last equality with (\ref{X}), we get  (\ref{L1}).

\end{proof}

\begin{lemma}
The following commutation relations take place:
\begin{equation}\label{L2}
[V_a,V_b]^i\approx C_{ab}^dV^i_d+D_{ab}^\alpha R^i_\alpha +J_{aj}W_b^{ji}-J_{bj}W_a^{ji}\,,
\end{equation}
where $C_{ab}^d$ is given by (\ref{LAD1}) and $W_a^{ji}=W_a^{ij}$.
\end{lemma}

\begin{proof} Applying the vector fields $V_c$ to (\ref{LAD1}) yields
$$
V_cV_aT_b-V_cV_bT_a\approx C_{ab}^dV_cT_d\,.
$$
Antisymmetrizing in indices $a,b,c$, we get
\begin{equation}\label{VVT}
\begin{array}{c}
[V_b,V_c]T_a+[V_c,V_a]T_b+[V_a,V_b]T_c
\approx C_{ab}^dV_cT_d+C_{bc}^dV_aT_d+C_{ca}^dV_bT_d
\\[3mm]
\approx C_{ab}^dV_dT_c+C_{bc}^dV_dT_a+C_{ca}^dV_dT_b\,.
\end{array}
\end{equation}
Introducing the vector fields
\begin{equation}\label{XX}
X_{ab}=[V_a,V_b]-C_{ab}^dV_d\,,
\end{equation}
we can rewrite (\ref{VVT}) as
$$
X_{bc}T_a+X_{ca}T_b+X_{ab}T_c\approx 0\,.
$$
Due to the regularity assumptions, the general solution to the last equation reads
\begin{equation}\label{XD}
X_{ab}^i\approx D_{ab}^\alpha R_\alpha^i+J_{aj}W^{ji}_b-J_{bj}W^{ji}_a
\end{equation}
for some coefficients $D_{ab}^\alpha$ and $W_a^{ji}=W_a^{ij}$. Comparing (\ref{XX}) with (\ref{XD}),  we arrive at (\ref{L2}).
\end{proof}

In the proofs above we have used the assumption of regularity of the field equations. Although this assumption guarantees that the commutators $[V_a, V_b]$ and $[R_\alpha, V_a]$ are given by  linear  combinations of the operators $V$, $R$, $J$, the coefficients $D$'s and $W$'s of these  combinations  may well be nonlocal\footnote{The $C$'s are local by definition.}. Locality of these coefficients will be our last assumption. It is automatically satisfied for the so-called \textit{integrable Lagrange structures} as they were defined in \cite{KLS2}.  We will not dwell on the concept of integrability of the Lagrange anchors referring the interested reader to the cited  paper. The only point we would like to mention here is that the integrability of the Lagrange structure is generally a stronger condition than the locality of the structure functions $D$'s and $W$'s.

\section{Covariant Poisson brackets}

In this section, we will continue to treat the gauge symmetry generators and the Lagrange anchor as being given by the collections of vector fields $R_\alpha$ and $V_a$ on $\mathcal{M}$; in so doing, the components of the field equations $T_a$, as well as the other sections associated with the system, will be regarded as globally defined functions on $\mathcal{M}$. With these conventions we can apply to them the usual formulae from differential geometry. In particular, introducing the exterior differential
$$
\delta=\delta \varphi^i \wedge \frac{\delta}{\delta \varphi^i}\,,
$$
we can identify the Jacobi operators with the components of the exact 1-forms
$$
\delta T_a=J_{ai}\delta \varphi^i\,.
$$
The property of a functional $A$ to be gauge invariant is expressed by the relation
\begin{equation}\label{RAAT}
R_\alpha A=A_\alpha^aT_a
\end{equation}
where the left hand side involves the action of (variational) vector fields on the functional.

\subsection{Brackets} The cornerstone of our construction is an \textit{advanced/retarded  fluctuation} $V_A^{\pm}$ caused by a physical observable $A$. By definition, $V_A^{\pm}$ is a vector field from $\Gamma_{\pm} (T\mathcal{M})$ satisfying the condition
\begin{equation}\label{fluct}
V_A^{\pm} T_a\approx V_a A\,.
\end{equation}
  We claim that this equation  defines $V_A^{\pm}$ uniquely up to adding a vector field from $\mathcal{R}$ and on-shell vanishing terms. In order to prove this assertion consider the on-shell commutative diagram with on-shell exact rows
$$
\xymatrix@C=0.5cm{
  \cdots \ar[r] & \Gamma_{\pm}(\mathcal{F}) \ar[rr]^{R} && \Gamma_{\pm}(T\mathcal{M}) \ar[rr]^{J} && \Gamma_{\pm}(\mathcal{E}) \ar[rr]^{L} && \Gamma_{\pm}(\mathcal{G}) \ar[r] & \cdots\\
  \cdots \ar[r] & \Gamma_{\pm}(\mathcal{G}^\ast)\ar[u]_{U} \ar[rr]^{L^\ast} && \Gamma_{\pm}(\mathcal{E}^\ast) \ar[u]_{V}\ar[rr]^{J^\ast} && \Gamma_{\pm}(T^\ast\mathcal{M})\ar[u]_{V^\ast} \ar[rr]^{R^\ast} && \Gamma_{\pm}(\mathcal{\mathcal{F}^\ast}) \ar[u]_{U^\ast}\ar[r] & \cdots
  }
$$
The 1-form $\delta A\in \Gamma_{\pm}(T^\ast \mathcal{M})$, being the differential of a physical observable, belongs to the on-shell kernel of the map $R^\ast$. From the on-shell commutativity of the right square it follows that $V_aA$ belongs to the on-shell kernel of the operator $L$, that is, $L_A^a(V_aA)\approx U_A^\alpha R_\alpha A\approx 0$. The on-shell exactness of the top sequence at $\Gamma_\pm(\mathcal{E})$ ensures the existence of a section $V^{\pm}_A\in \Gamma_{\pm}(T\mathcal{M})$ obeying  (\ref{fluct}). Furthermore, any two such sections  defer on shell by an element from  $\mathrm{Ker} J=\mathrm{Im} R$, that is, by a linear combination of the gauge symmetry generators.

Now we define the advanced/retarded Poisson brackets of two physical observables by the relation
\begin{equation}\label{PB}
\{A,B\}^{\pm}=V^{\pm}_AB-V^{\pm}_BA\,, \qquad \forall A, B\in \Phi^{\mathrm{inv}}_0\,.
\end{equation}
These brackets are well defined on shell as the ambiguity related to  the choice of the fluctuations,
\begin{equation}
V_A^\pm\quad \rightarrow\quad V_A^\pm+\xi^\alpha R_\alpha+T_a X^a \,,\qquad\xi\in \Gamma_\pm(\mathcal{F})\,,\quad X^a\in \Gamma_{\pm}(T\mathcal{M})\,,
\end{equation}
results in on-shell vanishing terms. It is also clear that the brackets are antisymmetric and bi-linear over $\mathbb{R}$. In order to prove the other properties of the Poisson brackets as well as the fact that the functional  (\ref{PB}) is a physical observable itself we need some further properties of the fluctuations.

The departure point for deriving these properties is the defining relation (\ref{fluct}) written in the off-shell form
\begin{equation}\label{fluct-off}
V_a A-V_A^{\pm} T_a+ A_a^cT_c=0\,.
\end{equation}
Here $A_a^b$ are some coefficients determined by $A$. Applying $V_b$ to both sides of the above relation and antisymmetrizing in $a$ and $b$, we find
 $$
0=[V_b,V_a] A-V_b V_A^{\pm} T_a+V_a V_A^{\pm} T_b+V_b (A_a^cT_c)-V_a (A_b^cT_c)
$$
$$
\approx [V_b,V_a] A+[V_a,V_A^{\pm}] T_b-[V_b, V_A^{\pm}] T_a+V^{\pm}_A(V_aT_b-V_bT_a) +[V_a,V_A^{\pm}] T_b+ A_a^cV_bT_c- A_b^cV_aT_c
$$
$$
\approx C_{ab}^d(V_d A-V_A^{\pm}T_a)+J_{bj}W_a^{ji}A,_i-J_{aj}W_b^{ji}A,_i-[V_b, V_A^{\pm}] T_a +[V_a,V_A^{\pm}] T_b+ A_a^cV_cT_b- A_b^cV_cT_a
$$
$$
\approx [V_a, V_A^{\pm}] T_b -[V_b,V_A^{\pm}] T_a +J_{bi}W_a^{ij}A,_j-J_{ai}W_b^{ij}A,_j+ A_a^cV_cT_b- A_b^cV_cT_a\,.
$$
Here we have used Rel. (\ref{L2}). Introducing the vector fields $Y_a$ with components
$$
Y^i_a=[V_a,V^\pm_A]^i+A_a^bV^i_b+W_a^{ij}A,_j\,,
$$
we can rewrite the last relation as
 $$
 Y_aT_b-Y_bT_a\approx 0\,.
 $$
 Due to the regularity of the field equations the vector fields $Y_a$ have the form
 $$
 Y^i_a=W^{\pm ij}_AJ_{aj}+D^\alpha_a R_\alpha^i
 $$
 for some $W^{\pm ij}_A=W^{\pm ji}_A$ and $D_a^\alpha$. The signs ``$\pm$'' indicate the support properties of $W_A^{\pm ij}$ in either index $i$ and $j$.  We call $W_A^{\pm ij}$ the \textit{secondary advanced/retarded fluctuation} of $A$. Thus,
\begin{equation}\label{VVA}
[V_a,V^\pm_A]^i\approx J_{aj}W^{\pm ji}_A -A,_jW_a^{ji}-A_a^bV^i_b +D^\alpha_a R_\alpha^i
\end{equation}
for any physical observable $A$.

In a similar manner, applying the gauge generators  $R_\alpha$ to (\ref{fluct-off}) and using (\ref{L1}), we obtain
$$
0\approx R_\alpha (V_aA-V^{\pm}_AT_a)\approx [R_\alpha,V_a]A+V_aR_\alpha A-[R_\alpha,V_A^\pm]T_a-V_A^\pm R_\alpha T_a
 $$
 $$
 \approx C_{\alpha a}^bV_bA +J_{aj}W^{ji}_\alpha A,_i+V_a(A_\alpha^bT_b) -[R_\alpha, V_A^\pm]T_a-V_A^\pm(C_{\alpha a}^bT_b)
 $$
 $$
 \approx C_{\alpha a}^b(V_bA-V_A^{\pm}T_b) +J_{aj}W^{ji}_\alpha A,_i+A_\alpha^b V_aT_b-[R_\alpha, V_A^\pm]T_a
 $$
 $$
 \approx A,_jW^{ji}_\alpha J_{ai}+A_\alpha^b V_bT_a-[R_\alpha, V_A^\pm]T_a\,.
 $$
 The last relation has the form $
 Y_\alpha T_a\approx 0$, where
 $$ Y^i_\alpha=A,_jW^{ji}_\alpha +A_\alpha^aV^i_a -[R_\alpha, V_A^{\pm}]^i\,.
 $$
 So, the vector fields $Y_\alpha$ generate gauge symmetry transformations and we can expand them as $Y_\alpha =A_\alpha^\beta R_\beta$. Finally, we get
 \begin{equation}\label{RVA}
 [R_\alpha, V_A^{\pm}]^i\approx A,_jW^{ji}_\alpha +A_\alpha^aV^i_a -A_\alpha^\beta R_\beta^i
 \end{equation}
 for any physical observable $A$.

 \begin{proposition} The brackets (\ref{PB}) map physical observables to physical observables.
 \end{proposition}

 \begin{proof}
 We only need to show that the brackets of two physical observables is a gauge invariant functional. We find
 $$
 R_\alpha \{A,B\}^{\pm}=R_\alpha V_A^\pm B-R_\alpha V^\pm_BA
 $$
 $$
 \approx [R_\alpha, V_A^\pm]B+V_A^\pm R_\alpha B-[R_\alpha, V_B^\pm]A-V_B^\pm R_\alpha A
 $$
 $$
 \approx A,_jW^{ji}_\alpha B,_i +A_\alpha^aV_aB+V_A^\pm(B_\alpha^a T_a)-B,_iW^{ij}_\alpha A,_j-B_\alpha^aV_aA-V_B^\pm(A_\alpha^aT_a)
 $$
 $$
 \approx B_\alpha^a(V_A^\pm T_a-V_aA)-A_\alpha^a(V_B^\pm T_a-V_aB)\approx 0\,.
 $$
 Here we used relations (\ref{RAAT}), (\ref{RVA}) and (\ref{fluct}).
 \end{proof}

Since the brackets of two physical observables $A$ and $B$ is again a physical observable one can ask about the explicit form of the advanced/retaded fluctuation caused by $\{A,B\}^{\pm}$. The answer is given by the next
 \begin{proposition}\label{fl}
 The advanced/retarded fluctuation caused by the brackets of two physical observables $A$ and $B$ is given by
 $$
 V_{\{A,B\}}^{\pm i}=[V^{\pm}_A,V^\pm_B]^i+B,_jW^{\pm ji}_A-A,_jW_B^{\pm ji}\,.
 $$
 \end{proposition}

\begin{proof} Using (\ref{fluct-off}) and (\ref{VVA}), we find
$$
V_a\{A,B\}^\pm=V_a(V_A^{\pm}B-V_B^\pm A)=[V_a,V_A^\pm]B+V_A^\pm V_aB-[V_a,V_B^\pm]A-V_B^\pm V_aA
$$
$$
\approx J_{ai}W_A^{\pm ij}B,_j-A,_iW^{ij}_aB,_j +A_a^bV_bB +V^\pm_A(V_B^\pm T_a-B_a^bT_b)
$$
$$- J_{ai}W_B^{\pm ij}A,_j+B,_iW^{ij}_aA,_j -B_a^bV_bA-V^\pm_B(V_A^\pm T_a-A_a^bT_b)\,.
$$
$$
\approx J_{ai}W_A^{\pm ij}B,_j- J_{ai}W_B^{\pm ij}A,_j+A_a^b(V_bB-V_B^\pm T_b)-B_a^b(V_bA-V_A^\pm T_b) +[V^\pm_A,V_B^\pm] T_a
$$
$$
\approx ([V^\pm_A,V_B^\pm]^i+B,_jW_A^{\pm ji}- A,_jW_B^{\pm ji})J_{ai}\,.
$$
It remains to compare the last relation with (\ref{fluct}).
\end{proof}
 \begin{proposition}
 The brackets (\ref{PB}) satisfy the Jacobi identity, that is,
 $$
 \{\{A,B\}^\pm, C\}^\pm+\{\{B,C\}^\pm, A\}^\pm+\{\{C,A\}^\pm,B\}^\pm\approx 0
 $$
 for any physical observables $A$, $B$ and $C$.
 \end{proposition}

 \begin{proof} Here we can use the following auxiliary construction. As is well known the total space of the cotangent bundle $T^\ast \mathcal{M}$ carries the canonical symplectic structure. Let $\bar\varphi_i$ denote  linear coordinates in the fibers of $T^\ast \mathcal{M}$. Then the canonical symplectic structure on $T^\ast \mathcal{M}$ determines and is determined by the following Poisson brackets:
 $$
 \{\varphi^i,\varphi^j\}=0\,,\qquad \{\bar\varphi_i,\varphi^j\}=\delta_i^j\,,\qquad \{\bar\varphi_i,\bar\varphi_j\}=0\,.
 $$
 In the present field-theoretical context one can think of the variables  $\bar\varphi_i$ as the sources for the fields $\varphi^i$.

 Now to any physical observable $A$ we can associate a pair of functions on $T^\ast \mathcal{M}$, namely,
\begin{equation}\label{AVW}
A^\pm =A+V_A^{\pm i}\bar\varphi_i+\frac12 W_A^{\pm ij}\bar\varphi_i\bar\varphi_j\,.
\end{equation}
The Jacobi identity for the canonical Poisson brackets on $T^\ast \mathcal{M}$ implies that
$$
\{\{A^\pm,B^\pm\}, C^\pm\}+cycle(A,B,C)=0\,.
$$
In particular,
$$
\{\{A^\pm,B^\pm\}, C^\pm\}|_{\bar\varphi=0}+cycle(A,B,C)=0\,.
$$
But
$$
\{\{A^\pm,B^\pm\}, C^\pm\}|_{\bar\varphi=0}\approx \{\{A,B\}^{\pm}, C\}^\pm\,.
$$
The last relation follows immediately from Proposition \ref{fl}.
 \end{proof}

\begin{proposition} Leibniz's rule holds for the brackets (\ref{PB}).

\end{proposition}

\begin{proof}
We first find the fluctuation caused by the product of physical observables. We have
$$
V_a(AB)=AV_aB+BV_aA\approx AV_B^{\pm}T_a+BV_A^\pm T_a\,.
$$
Hence,
$$
V^\pm_{AB}=AV_B^\pm+BV_A^\pm
$$
and
$$
\{AB,C\}=V_{AB}^\pm C-V_C^\pm(AB)=AV_B^\pm C+BV_A^\pm C-B V_C^\pm A-AV_C^\pm B
$$
$$
= A\{B,C\}+B\{A,C\}\,.
$$
\end{proof}

\subsection{Reciprocity relations} Using the concept of an advanced/retarded fluctuation, we have endowed the algebra of physical observables with the pair of Poisson brackets. One would expect the two Poisson structures to be closely related to each other as they originate from one and the same system of field equations and the Lagrange anchor. In this subsection, we will establish a precise link between both the brackets and compare our construction with the Peierls formula in the case of Lagrangian dynamics.

Let $A$ be a physical observable. Then $\delta A\in \Gamma_0(T^\ast\mathcal{M})$ and  $R_\alpha A\approx 0$. Therefore the section $\delta A$ belongs to the on-shell kernel of the operator $R^\ast: \Gamma_{\pm}(T^\ast \mathcal{M})\rightarrow \Gamma_{\pm}(\mathcal{F}^\ast)$. Exactness of the advanced/retarded sequence (\ref{ExSeq*}) implies then existence of an advanced/retaded section $A_{\pm}\in \Gamma_{\pm}(\mathcal{E}^\ast)$ such that
\begin{equation}\label{rep}
\delta A\approx A^a_{\pm}\delta T_a\,.
\end{equation}
We call $A_{\pm}$ the \textit{advanced/retarded repercussion} of $A$. Notice that relation (\ref{rep}) defines $A_{\pm}$ only up to adding an on-shell vanishing section or a section belonging to $\mathrm{Im } L^\ast $. This gives rise to the following equivalence relation on the space of all repercussions associated with $A$:
\begin{equation}\label{eq-rep}
A_{\pm}^{'a}\sim A_{\pm}^a \qquad \Leftrightarrow\qquad A_{\pm}^{'a} - A_{\pm}^a\approx L^a_{A}M^A
\end{equation}
for some $M\in \Gamma_{\pm}(\mathcal{G}^\ast)$. The exactness of (\ref{ExSeq*}) at $\Gamma_\pm (\mathcal{E}^\ast|_\Sigma)$ implies that each physical observable admits a unique advanced/retarded repercussion  modulo equivalence (\ref{eq-rep}).

Given a physical observable $A$, consider the vector fields $A_\pm^a V_a\in \Gamma_\pm (T\mathcal{M})$. In view of (\ref{LAD1}) we have
\begin{equation}
A_\pm^aV_a T_b\approx A_\pm^a V_bT_a\approx V_b A\,.
\end{equation}
The last relation is nothing else but the definition of the fluctuation caused by $A$. So, we can set $$V_A^\pm=A_\pm^a V_a\,.$$ Notice that the equivalence relations for  the repercussions (\ref{eq-rep}) and the fluctuations (\ref{fluct}) are compatible to each other due to relation (\ref{VL=RW}).

Using the definitions of the Lagrange anchor and the advanced/retarded repercussion, we get
\begin{equation}\label{BA-AB}
V_A^-B=A^a_- V_a^iB,_i \approx A^a_-V_a^iJ_{bi}B^b_+\approx A^a_-V_b^iJ_{ai}B^b_+ \approx B^b_+V^i_bA,_i=V^+_BA\,.
\end{equation}
Care is required in handling expressions with several contracted indices:
Although one can use either of the two repercussions $B_{\pm}$ to represent the functional derivative $B,_{i}$, only the advanced one is allowable in the middle equalities  (\ref{BA-AB}). That choice ensures that all the integrals converge properly and we can freely change the order of contractions.
Ignoring convergence, one could erroneously conclude that $V_A^-B\approx V_B^-A$ for all $A, B\in \Phi_0^{\mathrm{inv}}$, while the true equality is
\begin{equation}\label{RR1}
V_A^-B\approx V_B^+A\,.
\end{equation}
In the case of Lagrangian equations endowed with the canonical Lagrange anchor, Eq. (\ref{RR}) is known as the \textit{reciprocity relation} for physical observables \cite{DW}. It just says that  the retarded effect of $A$ on $B$ equals the advanced effect of $B$ on $A$, and vice versa.

As an immediate consequence of (\ref{RR1}) we obtain the equality
\begin{equation}
\{A,B\}^+=-\{B,A\}^-\,.
\end{equation}
So, using the advanced or the retarded fluctuations yields essentially the same Poisson brackets that differ only by sign. The appearance of the minus sign has a clear physical explanation: interchanging past and future one in fact reverses the direction of time, so that the corresponding Hamiltonian equations of motion have to acquire an overall minus sign.

Using the advanced and retarded repercussions of physical observables, we can rewrite the  Poisson brackets (\ref{PB}) in the form
\begin{equation}\label{PB1}
\{A,B\}^\pm=\pm {\tilde{V}}_A B\,,
\end{equation}
where
\begin{equation}\label{VA}
\tilde{V}_A=\tilde{A}^aV_a\,,\qquad \tilde{A}=A_+-A_-\,.
\end{equation}
We call $\tilde{A}\in \Gamma(\mathcal{E}^\ast)$ the \textit{characteristic} of the physical observable $A$ and refer to  $\tilde{V}_A\in \Gamma(T\mathcal{M})$ as the \textit{causal fluctuation} produced by $A$.  From the definitions (\ref{rep}) and (\ref{VA}) it follows  that
 \begin{equation}\label{TA}
 \tilde{A}^a \delta T_a\approx 0\,
 \end{equation}
for any physical observable $A$.
 We see that the section $\tilde{A}$ behaves like the (pre)characteristic of a conservation law \cite{KLS1}, hence the name.  The crucial distinction of $\tilde{A}$ from the true characteristic is that its components $\tilde{A}^a$ are  not generally local functions of fields and their derivatives. So, one can't expect the functional  $\tilde{A}^aT_a$ to be given on shell by the integral of the total divergence of  a conserved current.

From (\ref{TA}) it also follows that
$$
\tilde{V}_A T_a\approx 0\,.
$$
In other words, the vector fields $\tilde{V}_A$ are tangent to the shell $\Sigma$ and generate the transformations of the space $\mathcal{M}$ that moves solutions of the field equations to solutions.  Again, the main difference  between causal fluctuations and the usual generators of gauge or global symmetries is that the components of the vector field $\tilde{V}_A$ are not generally local functions of fields and their derivatives.

Yet another form of the Poisson brackets (\ref{PB}) is given by
$$
\{A,B\}^{\pm}=\pm (A_+^a G_{ab}B_-^b -A_-^aG_{ab}B_+^b)\,,
$$
where $G$ is the generalized Van Vleck operator (\ref{VV}). Notice that the support properties of the repercussions on the right ensure the convergence of all the integrals.

Let us now compare the Poisson brackets (\ref{PB}) with the usual Peierls' brackets in the Lagrangian field theory. In the latter case the dynamics of fields are governed by some action functional $S[\varphi]$. As was explained in Sec. 3, the corresponding equations of motion $S,_i[\varphi]=0$ admit the canonical Lagrange anchor given by the unit operator $V=1$ on $\Gamma(T^\ast \mathcal{M})$. The definition of the advanced/retarded fluctuation (\ref{fluct}) takes the form
\begin{equation}\label{V=1}
V_A^{\pm i}S,_{ij}\approx A,_j\,.
\end{equation}
In the absence of gauge symmetries this equation can be solved for $V_A^{\pm}$ with the help of the advanced/retarded Green function $G^{\pm ij}$. By definition,
\begin{equation}\label{GF}
G^{\pm in}S,_{nj}=S,_{jn}G^{\pm ni}=\delta^i_j\quad \mbox{and}
\quad G^{- ij}= 0=G^{+ji} \quad \mbox{if}\quad j>i\,.
\end{equation}
Here $j>i$ means that the time associated with the index $i$ lies to the past of the time associated with the index $j$. Besides (\ref{GF}),  the advanced and retarded Green functions satisfy the so-called \textit{reciprocity relation}
\begin{equation}\label{RR}
  \qquad G^{\pm ij}= G^{\mp ji}\,.
\end{equation}
In terms of the Green functions  the advanced/retared solution to (\ref{V=1}) reads
\begin{equation}\label{VGA}
  V_A^{\pm i}=G^{\pm ij}A,_j\,.
\end{equation}
and the causal fluctuation takes the form $\tilde{V}^i_A=V^{+i}_A-V^{-i}_A=\tilde{G}^{ij}A,_j$, where the difference $\tilde{G}=G^+-G^-$ is known as the causal Green function.  In view of the reciprocity relation (\ref{RR}), $\tilde{G}^{ij}=-\tilde{G}^{ji}$.  Substituting  (\ref{VGA}) into (\ref{PB1}), we get
\begin{equation}\label{AGB}
  \{A,B\}^\pm=\pm A,_i\tilde{G}^{ij}B,_j\,.
\end{equation}
The antisymmetry of the brackets as well as the derivation property are obvious. The direct verification of the Jacobi identity for (\ref{AGB}) can be found in \cite{DW}, \cite{BEMS}.
For explicit calculations of causal Green functions on curved backgrounds see e.g. \cite{ER}.

In the presence of gauge symmetry the field equations admit no advanced/retarded Green function (\ref{GF}) since  the corresponding Jacobi operator $J_{ij}=S,_{ij}$ is degenerate. As a result the fluctuations caused by the physical observables are not uniquely determined by equation (\ref{V=1}). To avoid this ambiguity and obtain a particular solution for the fluctuation one usually imposes supplementary  constraints  on the fields, $\chi^\alpha[\varphi]=0$, called the gauge fixing conditions. These conditions are required to be chosen in such a way that the operator
$$
\Delta^\beta_\alpha=R_\alpha \chi^\beta\,,
$$
called the Faddeev-Popov operator, is invertible, i.e., has Green's functions. Then, with account of the gauge fixing conditions the original system of equations can be transformed into an equivalent system of hyperbolic partial differential equations with constraints, for which the advanced and retarded Green functions can be defined. A more detailed discussion  of the gauge fixing procedure\footnote{Usually, this procedure  is discussed only for the Lagrangian field equations, but in principle it works quite generally.} can be found in \cite{DW}, \cite{Kh}. Let me  stress that in our treatment of the Poisson algebra above we did not perform any explicit gauge fixing.

\section{Examples}

\subsection{The Pais-Uhlenbeck oscillator}

The  PU oscillator is described by the fourth-order differential equation
\begin{equation}\label{PU}
\left(\frac{d^2}{dt^2}+\omega^2_1\right)\left(\frac{d^2}{dt^2}+\omega_2^2\right)x=0\,,
\end{equation}
where the constants $\omega_{1}$ and $\omega_2$ have the meaning of frequencies.
The advanced/retarded Green function $G^\pm (t_2-t_1)$ for (\ref{PU}) is given by
$$
G^\pm(t)=\pm \frac{\theta(\mp t)}{\omega_2^2-\omega_1^2}\left(\frac{\sin \omega_1t}{\omega_1}-\frac{\sin\omega_2t}{\omega_2}\right)\,.
$$
Here $\theta(t)$ is the step function:
$$
\theta(t)=\left\{
            \begin{array}{ll}
              $1$, & \hbox{for $t>0$;} \\[2mm]
              $0$, & \hbox{for $t<0$.}
            \end{array}
          \right.\,, \qquad \frac{d\theta}{dt}(t)=\delta(t)\,.
$$
The differential equation (\ref{PU}) admits the two-parameter family of the Lagrange anchors \cite{KLS3}
\begin{equation}\label{VPU}
V= \alpha+\beta \frac{d^2}{dt^2}\,,\qquad \alpha,\beta \in \mathbb{R}\,.
\end{equation}
In this particular case the defining condition for the Lagrange anchor (\ref{LAD1}) reduces to the commutativity of the operator $V$ with the fourth-order differential operator defining the equation of motion  (\ref{PU}). The operators obviously commute as any pair of differential operators with constant coefficients. Notice also that the equation of motion (\ref{PU}) is Lagrangian and the canonical Lagrange anchor corresponds to $\alpha=1$, $\beta=0$.

The advanced Poisson brackets are given by
$$
\{x(t_1), x(t_2)\}^+= V\tilde{G}(t_1-t_2)
$$
$$=\left(\frac{\alpha-\beta\omega^2_1}{{\omega_2^2-\omega_1^2}}\right)\frac{\sin\omega_1(t_1-t_2)}{\omega_1}
-\left(\frac{\alpha-\beta\omega^2_2}{{\omega_2^2-\omega_1^2}}\right)\frac{\sin\omega_2(t_1-t_2)}{\omega_2}\,.$$

Differentiating by $t_1$, $t_2$ and setting $t_1=t_2$, we obtain the following  Poisson brackets of the phase-space variables $z=(x, \dot x, \ddot x, \dddot x)$:
\begin{equation}\label{PU-PB}
\begin{array}{c}
\{\dot x, x\}^+=\beta\,,\qquad \{\dot x, \ddot x\}^+=   \{\dddot x, x\}^+=\alpha-\beta(\omega_1^2+\omega_2^2)\,,\\[3mm] \{\ddot x, \dddot x\}^+=\alpha(\omega_1^2+\omega_2^2)-\beta(\omega_1^4+\omega_1^2\omega_2^2+\omega_2^4)\,,
\end{array}
\end{equation}
and the other brackets vanish. For $\alpha=1$, $\beta=0$, this yields  the standard Poisson brackets on the phase space of the PU oscillator.

The Pfaffian of the Poisson bi-vector is given by
$$
\beta^2\omega_1^2\omega_2^2-\alpha\beta(\omega_1^2+\omega_2^2)+\alpha^2\,.
$$
It vanishes when $\alpha=\beta=0$ or $\alpha/\beta =\omega_{1,2}^{2}$. For all other values of $\alpha$ and $\beta$ the Poisson brackets (\ref{PU-PB}) are non-degenerate.

With the Poisson brackets (\ref{PU-PB}) the equations of motion (\ref{PU}) can be written in the Hamiltonian form
$$
\dot z^i=\{H, z^i\}^+\,, \qquad i=1, 2, 3, 4\,,
$$
where the Hamiltonian is given by
$$
H=\frac12\frac{(\dddot x+\omega_1^2\dot x)^2+\omega_2^2(\ddot x+\omega_1^2 x)^2}{(\omega_1^2-\omega_2^2)(\alpha-\beta\omega^2_2)}-
\frac12\frac{(\dddot x+\omega_2^2\dot x)^2+\omega_1^2(\ddot x+\omega_2^2 x)^2}{(\omega_1^2-\omega_2^2)(\alpha-\beta\omega^2_1)}\,.
$$
As was first noticed in \cite{KLS3}, this Hamiltonian is positive definite (i.e., determined by a positive definite quadratic form on the phase space), whenever
\begin{equation}\label{ineq}
\omega_1^2>\frac{\alpha}{\beta}>\omega_2^2\,.
\end{equation}
Clearly, the canonical Lagrange anchor $(\alpha=1, \beta=0)$ does not satisfy these inequalities for any frequencies $\omega_{1,2}$. On the other hand, in the absence of resonance $(\omega_1\neq \omega_2)$, one can always choose a non-canonical Lagrange anchor (\ref{VPU}) to meet the inequalities (\ref{ineq}). Upon quantization the positive-definite Hamiltonian will have a positive energy spectrum and a well-defined ground state.  The last property is crucial for the quantum stability of the system \cite{KLS3}.
\subsection{Chiral bosons in two dimensions} Consider a multiplet of self-dual 1-forms $\phi^i$ in two-dimensional Minkowski space.  The field equations are given by the closedness condition $d\phi^i=0$. In terms of the light-cone coordinates $x^{\pm}=\tau\pm\sigma $, we have $\phi^i=\phi_+^idx^+$ and the equations of motion take the form
\begin{equation}\label{ChB}
\partial_-\phi^i_+=0\,.
\end{equation}
These equations are clearly non-Lagrangian. The advanced, retarded and causal Green functions for (\ref{ChB}) read
$$
G_i^{j \pm}(x)=\pm \theta (\pm x^-)\delta(x^+)\delta_i^j\,,\qquad \tilde{G}_i^j(x)=\delta(x^+)\delta_i^j\,.
$$
Notice that the operator $\tilde{G}$ is symmetric.

As was shown in \cite{KLS1}, the field equations (\ref{ChB}) admit the following family of the  Lagrange anchors:
\begin{equation}\label{Vij}
V^{ij}=\kappa g^{ij}\partial_+ +f^{ij}_k\phi_+^k(x)\,,\qquad \kappa\in \mathbb{R}\,.
\end{equation}
Here $f^{ij}_k$ are the structure constants of a semi-simple Lie algebra $\mathcal{G}$ and $g^{ij}$ are the components of the Killing form on $\mathcal{G}$. As is seen from (\ref{Vij}), the operator $V$ is antisymmetric. According to the general definition (\ref{PB1}), the advanced Poisson brackets of fields are  given by
$$
\{\phi^i_+(x),\phi_+^j(\tilde{x})\}^+=V^{ik}\tilde{G}_k^j(x-\tilde{x})=f^{ij}_k\phi_+^k(x)\delta(x^+-\tilde{x}^+)+\kappa g^{ij}\delta'(x^+-\tilde{x}^+)\,.
$$
In the equal time limit $\tau=\tilde{\tau}$ we get
$$
\{\phi^i_+(\sigma),\phi_+^j(\tilde{\sigma})\}^+=f^{ij}_k\phi_+^k(\sigma)\delta(\sigma-\tilde{\sigma})+\kappa g^{ij}\delta'(\sigma-\tilde{\sigma})\,.
$$
Thus, the equal-time brackets of the fields $\phi^i_+$ define the affine Lie algebra $\hat{\mathcal{G}}$ of level $\kappa$.

Using these Poisson brackets, one can rewrite the field equations (\ref{ChB}) in the Hamiltonian from
$$
\partial_\tau \phi^i_+=\{H, \phi^i_+\}^+
$$
with respect to the Hamiltonian
$$
H=-\frac1{2\kappa}\int g_{ij}\phi_+^i\phi_+^j d\sigma\,.
$$
For a compact  Lie algebra $\mathcal{G}$ this Hamiltonian can be made positive definite.
\subsection{Maxwell's electrodynamics} Let $\Lambda=\bigoplus \Lambda^p$ denote the exterior algebra of differential forms on the $4$-dimensional Minkowski space. The space $\Lambda $ is endowed with the standard inner product
\begin{equation}\label{IP}
(A, B)=\int_{\mathbb{R}^{3,1}}A\wedge *B\,,
\end{equation}
where $\ast: \Lambda^p\rightarrow \Lambda^{4-p}$ is the Hodge operator with respect to the Minkowski metric. Notice that $\ast^2=-1$. The inner product, being non-degenerate in each degree,  allows us to identify the space $\Lambda^p$ with its dual vector space. We let $\delta=\ast d\ast$ denote the DeRham co-differential. The operator $\delta$ coincides, up to sign factor, with the formal adjoint of the exterior differential $d$ with respect to (\ref{IP}).

In the first-order formalism, the electromagnetic field is described by the strength tensor $F\in \Lambda^2$ subject to the Maxwell equations
\begin{equation}\label{T12}
T_1=\delta \ast F-I=0\,,\qquad T_2=\delta F-J=0\,.
\end{equation}
Here the 1-forms $I$ and $J$ represent the magnetic and electric currents, respectively. The self-consistency of the Maxwell equations implies that either of currents is conserved,
$$
\delta I=0\,,\qquad \delta J=0\,,
$$
and the equations satisfy the gauge identities
\begin{equation}\label{dT11}
\delta T_1\equiv 0\,,\qquad \delta T_2\equiv 0\,.
\end{equation}
These identities are clearly irreducible in four dimensions. At the same time, there is no gauge symmetry as all the components of the strength tensor are physically observable. This indicates that the field equations (\ref{T12}) are non-Lagrangian.

To make contact with the general definitions of Sec. 2, let us note that the configuration space of fields $\mathcal{M}$ is given here by the space $\Lambda^2$, the sections of the dynamics bundle assume their values in $\Lambda^1\oplus \Lambda^1$, and the bundle of gauge identities is described by sections with values in $\Lambda^0\oplus\Lambda^0$ . Since $\mathcal{M}$  is a linear space, we can identify the tangent space $T_F\mathcal{M}$ at each point $F\in \mathcal{M}$ with the space $\Lambda^2$ itself.

As was shown in \cite{KazLS}, the system (\ref{T12}) possesses the Lagrange anchor  $V=(V_1,V_2)$ defined by the relation
\begin{equation}\label{V12}
V_1=0\,,\qquad V_2[\Psi]=(d \Psi, D)\,,\qquad D=dx^\mu\wedge dx^\nu\frac{\delta}{\delta F_{\mu\nu}(x)}\,,
\end{equation}
where  $\Psi$ is a test 1-form considered as a section of the dual of the dynamics bundle.
In order to describe the covariant Poisson brackets associated with this Lagrange anchor  we introduce the linear functionals of the electromagnetic field
\begin{equation} \label{AF}
A[F]=(F,A)\,,
\end{equation}
where $A$ is an arbitrary 2-form with compact support. Let $A_\pm =(A^1_\pm, A^2_\pm)$ denote the advaneced/retarded repercussion of the physical observable  $A[F]$. By definition,  $A_\pm^{a}$, $a=1,2$, are 1-forms with advanced/retarded supports obeying the equation
\begin{equation}\label{AAA}
\ast dA_\pm^1+dA_\pm^2=A\,.
\end{equation}
In view of the gauge identities (\ref{dT11}) the last equation does not specify a unique repercussion: If $A^{a}_\pm$ is a solution to (\ref{AAA}), then $A^a_\pm+dB_\pm$ is again a solution for arbitrary functions $B_\pm^a$ with advanced/retarded supports. To fix this ambiguity we impose the supplementary conditions
$$
\delta  A_\pm^{a}=0\,,\qquad a=1,2\,.
$$
Then applying to both sides of equation ({\ref{AAA}}) the operators $\ast d$ and $\delta$, we obtain
$$
\square A_\pm^1=\ast d A\,,\qquad \square A_\pm^2= \delta A\,,
$$
where $\square=(\delta d+d\delta)$ is the d'Alambert operator.  The advanced/retarded Green functions for $\square$ read
$$
G^{\pm}(x)=\frac{1}{4\pi}\theta(\mp t)\delta(x^2)=\frac{1}{4\pi} \frac1{|\mathbf{x}|}\delta(t\pm |\mathbf{x}|)\,,\qquad x=(t,\mathbf{x})\,,
$$
and we can write
\begin{equation}\label{A1A2}
A_\pm^1=G^\pm\ast d A=\ast d (G^\pm A)\,,\qquad A_\pm^2= G^\pm\delta A=\delta  (G^\pm A)\,.
\end{equation}
Now substituting Rels. (\ref{V12}), (\ref{AF}), (\ref{A1A2}) into the general formulae (\ref{PB1}), (\ref{VA}), we arrive at the following Poisson brackets:
\begin{equation}\label{AFBF}
\{A[F],B[F]\}^\pm=V_2[A_\pm^2]B[F]-V_2[B_\pm^2]A[F]= \pm(\tilde{G}\delta A,\delta B)\,.
\end{equation}
Here the causal Green function is given by
$$
\tilde{G}(x)=G^+(x)-G^-(x)=\frac1{4\pi}\frac 1{|\mathbf{x}|}(\delta(t+|\mathbf{x}|)-\delta(t-|\mathbf{x}|))\,.
$$
From (\ref{AFBF}) we deduce the following Poisson brackets for the components of the strength tensor:
$$
\{F_{\alpha\beta}(x),F_{\mu\nu}(x')\}^+=(\eta_{\alpha\mu}\partial_\beta\partial_\nu-\eta_{\beta\mu}\partial_\alpha\partial_\nu
-\eta_{\alpha\nu}\partial_\beta\partial_\mu+\eta_{\beta\nu}\partial_\alpha\partial_\mu)\tilde{G}(x-x')\,.
$$
In terms of $3+1$ splitting of the coordinates the nonzero Poisson brackets of fields are given by
\begin{equation}\label{FF}
\{F_{0k}(0), F_{ij}({x})\}^+=(\delta_{kj}\partial_i-\delta_{ki}\partial_j)\partial_t\tilde{G}(x)\,.
\end{equation}
In order to obtain the equal-time Poisson brackets in the phase space we note that
$$
\lim_{t\rightarrow 0}\tilde{G}(x)=\lim_{t\rightarrow 0}\partial_t^2 \tilde{G}(x)=0\,,\qquad \lim_{t\rightarrow 0}\partial_t \tilde{G}(x)=-\frac{\delta'(|\mathbf{x}|)}{4\pi |\mathbf{x}|}\,.
$$
Taking into account the integration measure
$$
d^3\mathbf{x}=|\mathbf{x}|^2d|\mathbf{x}|d\Omega\,,
$$
we find
$$
-\frac{\delta'(|\mathbf{x}|)}{4\pi|\mathbf{x}|}=\delta^3(\mathbf{x})\,.
$$
In the equal-time limit Eq. (\ref{FF}) gives the standard Poisson brackets
$$
\{E^i(\mathbf{x}), H^j(\mathbf{x}')\}^+=-2\epsilon^{ijk}\partial_k\delta^3(\mathbf{x}-\mathbf{x}')
$$
for the $3$-vectors
$$
E^i=F^{0i}\,,\qquad H^i=\varepsilon^{ijk}F_{jk}
$$
of electric and magnetic fields.

\section{Conclusion}

In this paper, we have defined covariant Poisson brackets in the space of true histories, starting from the classical equations of motion and a compatible Lagrange structure. These Poisson brackets generalize the Peierls' bracket construction to the case of not necessarily Lagrangian theories. Contrary to the conventional Poisson brackets from Hamiltonian dynamics the covariant Poisson brackets satisfy the Jacobi identity only when restricted to the algebra  of gauge invariant functionals considered modulo equations of motion. The last property allows one to identify them with the \textit{weak Poisson brackets} introduced in Ref. \cite{LS0}. In that paper, we shown that, under certain regularity conditions, the weak Poisson brackets give rise to a flat $P_\infty$-structure on the space of functionals. The first structure map $P_1$ is given here by the Koszul-Tate differential associated with the dynamical shell $\Sigma$ and the gauge distribution $\mathcal{R}$, the second structure map $P_2$ is determined by the weak Poisson bi-vector itself, and the higher maps are systematically constructed by means of homological perturbation theory. Given the $P_\infty$-structure, one can apply Kontsevich's  formality theorem to perform a deformation quantization of the classical system. In the absence of quantum anomalies, the result is a flat $A_\infty$-algebra, whose second structure map $A_2$ defines a weakly associative $\ast$-product.  The physical observables are identified with the cohomology of the differential $A_1$ in the ghost number $0$.  The weakly associative $\ast$-product passes through  the cohomology, inducing an associative $\ast$-product in the space of physical observables. So, applying the formality map to the weak Poisson structures allows one, in principle, to perform a covariant deformation quantization of a classical theory endowed with an integrable Lagrange structure. Of course, there is a plenty of technical difficulties (related, for instance, to divergences) in adapting the Kontsevich $\ast$-product to field theory.

There is also another potential way of quantizing the covariant Poisson brackets above. As was shown in the original paper \cite{KazLS}, each Lagrange structure defines and is defined by a flat $S_\infty$-structure whose first structure map implements the Koszul-Tate resolution of the dynamical shell, while the second one involves the Lagrange anchor. This $S_\infty$-structure can be compactly described by a single generating function $\Omega$ called the BRST charge. The BRST charge is defined to be an odd local functional on the ghost-extended phase space of fields and sources. It obeys certain boundary conditions as well as the master equation $\{\Omega, \Omega\}=0$, which is similar to that of the BFV-BRST formalism. In particular, the commutation relations stated by Lemmas 3.1 and 3.2 are obtained by expanding the master equation in powers of ghost variables and sources. The physical observables are naturally identified with the BRST-invariant functionals of ghost number zero; in so doing, two physical observables are considered as equivalent if they differ by a BRST-exact functional. The covariant Poisson brackets (\ref{PB}) can now be re-derived in the BRST terms. The point is that any gauge invariant functional of the original fields can be lifted to a BRST-invariant functional on the ghost-extended phase space of fields and sources. (In a truncated form such a lift is represented by formula (\ref{AVW}) from the proof of the Jacobi identity.) The lift is not unique. Denoting by $\sigma_{\pm}(A)$ the advanced/retarded lift of an observable $A$, so that $\{\sigma_{\pm}(A), \Omega\}=0$, we can define the covariant Poisson bracket (\ref{PB}) by the rule
\begin{equation}\label{sigma}
\{A,B\}^{\pm}=\sigma_{\pm}^{-1}(\{\sigma_{\pm}(A),\sigma_{\pm}(B)\})\,,
\end{equation}
where the projection $\sigma_{\pm}^{-1}$ just sets all the sources and ghost variables to zero.
At the quantum level, the classical BRST charge $\Omega$ is replaced, in the absence of quantum anomalies, by a Hermitian operator $\hat{\Omega}$ obeying the nilpotency condition $\hat{\Omega}{}^2=0$. The algebra of quantum observables consists of the Hermitian operators $\hat{A}$ in ghost number zero that commute with the BRST operator, i.e., $[\hat{\Omega}, \hat{A}]=0$. Following the pattern above, one can try to interpret the last equation as defining a quantum lift $\hat{A}=\hat{\sigma}_\pm(A)$ of a classical observable $A$. An advanced/retarded $\ast$-product could then be defined by the formula $A\ast_{_\pm} B=\hat{\sigma}_\pm^{-1}(\hat{\sigma}_\pm(A)\hat{\sigma}_\pm(B))$, which is similar to (\ref{sigma}). Whether such a quantum lift $\hat{\sigma}_\pm$ and a compatible projection $\hat{\sigma}_\pm^{-1}$ can really be constructed is not clear at the moment. I am going to address  this problem elsewhere.

\vspace{0.2 cm}

\noindent \textbf{Acknowledgements.} I wish to thank Simon Lyakhovich for a fruitful collaboration at the early stage of this work.  The work was partially supported by the Tomsk State
University Competitiveness Improvement Program, the RFBR grant
13-02-00551 and  the Dynasty Foundation.

\end{document}